\documentclass[
    11pt,
    letterpaper,
    reprint,
    notitlepage,
    superscriptaddress,
    aps,prx
]{revtex4-2}

\usepackage{amsmath,amssymb,amsfonts} 

\usepackage{physics}

\usepackage{subdepth} 	

\usepackage{bm}	
\renewcommand{\mathbf}{\bm}
\usepackage{dsfont}	
\renewcommand{\mathbb}{\mathds}	

\usepackage[svgnames,dvipsnames]{xcolor}
\usepackage{graphicx}
\definecolor{NewBlue}{rgb}{0.1, 0.1, 0.7}
\definecolor{NewRed}{rgb}{0.7, 0.1, 0.1}
\usepackage[colorlinks,
    linkcolor=Maroon,
    citecolor=NewBlue,
    urlcolor=NewBlue]{hyperref}

\usepackage[capitalise]{cleveref}
\usepackage{leftindex}

\usepackage[normalem]{ulem}

\usepackage{amsthm} 

\usepackage{tikz}
\usetikzlibrary{arrows.meta, positioning}

\usepackage{pdfpages}
\makeatletter
\AtBeginDocument{\let\LS@rot\@undefined}
\makeatother

\newtheorem{lemma}{Lemma}
\newtheorem{theorem}{Theorem}

\renewcommand{\t}[1]{\mathrm{#1}}

\renewcommand{\rho}{\varrho}
\newcommand{\p}{\partial}

\newcommand{\K}{\mathcal{K}}
\newcommand{\J}{\mathcal{J}}
\newcommand{\EE}{\mathbb{E}}

\begin{document}

\title{Completely-positive non-signalling non-Markovian dynamics}

\author{Serhii Kryhin}
\email{skryhin@g.harvard.edu}
\affiliation{Department of Physics, Harvard University, Cambridge 02139 MA, USA}

\author{Vivishek Sudhir}
\affiliation{LIGO Laboratory \& Department of Mechanical Engineering, Massachusetts Institute of Technology, 
Cambridge 02139 MA, USA}

\date{\today}

\begin{abstract}
We define non-Markovian quantum dynamics as evolution in which the current state depends on 
all past states, and completely characterize its structure under the assumptions of 
complete positivity and non-signalling.
The resulting continuous-time dynamics is an integro-differential equation that augments the 
Gorini-Kossakowski-Sudarshan-Lindblad equation with a memory integral, 
and is capable of describing the quantum state of systems exposed to noise with 
any integrable power spectral density with no further approximations.
We then establish a formalism to evaluate multi-time correlations of 
measurement outcomes in this general setting, obviating the need for a regression theorem.
As an application, we derive the emission spectrum of a driven two-level system 
coupled to a non-Markovian bath: the familiar Mollow triplet 
acquires a frequency-dependent linewidth that encodes the memory of the bath.
Our work provides a rigorous yet transparent description of the quantum state 
of non-Markovian systems, 
opening the door for state estimation and state-based quantum control 
beyond the Markovian regime.
\end{abstract}

\maketitle


\emph{Introduction.}
Despite the alluring simplicity of Markovian models of physical processes,
non-Markovian processes are ubiquitous in both classical and quantum 
physics \cite{Kubo,Kam23,BP02}.
Notorious examples include creep, \cite{Tsch89}, crackle \cite{Seth01}, and glassy-ness \cite{Phill96}; anomalous diffusion \cite{BouGeo90}; and
``$1/f$'' noise in a wide range of systems, including classical \cite{Calo74,Ziel79} and 
quantum electronics \cite{PalAlt14,BrowBla15,MullLise19,RowOliv23,YonTar23}, 
mechanical oscillators used for precision sensing \cite{Attkinson1963,Saul90,Kuro99,GrobAsp15,FedKipp18,CripCorbitt19}, 
and a bewildering variety of natural phenomena \cite{Mand68,West90}.

Considerable effort has been put into understanding non-Markovian dynamics in the quantum setting (see Refs. \cite{RivPlen14,BreuVacc16,RevModPhys.89.015001,LiWise18,Chru22} for reviews). 
Techniques employed to describe non-Markovian quantum dynamics, when they are 
model-agnostic \cite{Nakajima,Zwanzig,Mori65,ShibHash79,Grahm89,Breuer2007,ChruKoss10,GioPalm12,
Polloc2018,White2025}, often invoke ad-hoc or uncontrolled approximations, or are model-specific in 
some way \cite{HuPazZha92,Bud04,Shabani2005,BreuVac08,Ciccarello2013,Vacchini2013,Vacch16,Lorenzo2016},
or, in the case of systems afflicted by ``$1/f$'' noise, phenomenological to the extent that it 
is unclear how to describe the quantum state of the system \cite{PalAlt14,BrowBla15,Saul90}.
That is, a precise and general characterization of non-Markovian quantum processes, as done
for Markovian processes by Gorini-Kossakowski-Sudarshan-Lindblad (GKSL) \cite{GKS76,Lind76}, 
is lacking.

We define and characterize a class of non-Markovian quantum processes that are 
completely-positive and non-signaling.
The non-signaling condition crucially constrains the structure of the non-Markovian contributions
with no further assumptions.
The continuous-time dynamics that we derive is a non-Markovian generalization of 
the GKSL equation,
and describes the state evolution of quantum systems driven by 
noise of any integrable power spectral density.
We show how statistics of measurement outcomes can be derived from this formalism even though
a regression theorem is unavailable.
Finally, we illustrate the formalism on a non-Markovian generalization 
of a driven two-level system \cite{Mollow1969}.

\emph{A class of non-Markovian processes.}
In classical probability theory, a non-Markovian process is, informally, a random process whose
current value depends on all of its past values. A classical dynamical system is said to be
non-Markovian if its state is a non-Markovian process, i.e. if its current state depends on all of its
past states. It is this notion of non-Markovianity that we wish to extend to the quantum setting.

To wit, we consider a quantum system whose density matrix $\hat \rho_n$, at each time $t_n$, fully determines all of its observables at that time. Given a sequence of states $\hat{\rho}_0,\ldots, 
\hat{\rho}_n$, we say that the system is non-Markovian if the state $\hat{\rho}_{n+1}$ depends on all
previous states, i.e. there exists a non-trivial $(n+1)-$variable map $\K_{n+1}$ such that 
\begin{equation}\label{eq:Kn}
	\hat \rho_{n+1} = \K_{n+1}[\hat \rho_n, .., \hat \rho_0].
\end{equation}
This approach presumes no specific form of the state-history and 
is distinct from approaches where non-Markovian dynamics is defined in terms of the history 
of control operations \cite{MilMod21} (rather than the history of states).
Our concern is to characterize the structure of the maps $\{\K_n\}$, 
in the continuous-time limit, under a minimal set of physically reasonable axioms.

\begin{figure}[b]
\centering
\resizebox{\linewidth}{!}{%
\begin{tikzpicture}[
    box/.style={draw, minimum width=1.4cm, minimum height=1.2cm, thick, font=\large},
    lbl/.style={font=\normalsize},
    arr/.style={-{Latex[length=2.5mm]}, thick},
]

\def\yw{0.35}  

\foreach \i in {0,1,2} {
    \pgfmathsetmacro{\xpos}{\i * 2.5}
    \node[box] (I\i) at (\xpos, 0) {$\mathcal{E}_{\i}$};
}


\draw[thick] (-1.8, \yw) -- (I0.west |- 0,\yw);
\draw[arr] (-1.5, \yw) -- (-1.05, \yw);
\node[lbl, above] at (-1.25, \yw) {$\hat\sigma_0$};

\draw[thick] (I0.east |- 0,\yw) -- (I1.west |- 0,\yw);
\draw[arr] (1.0, \yw) -- (1.45, \yw);
\node[lbl, above] at (1.25, \yw) {$\hat\sigma_1$};

\draw[thick] (I1.east |- 0,\yw) -- (I2.west |- 0,\yw);
\draw[arr] (3.5, \yw) -- (3.95, \yw);
\node[lbl, above] at (3.75, \yw) {$\hat\sigma_2$};

\draw[thick] (I2.east |- 0,\yw) -- (6.8, \yw);
\draw[arr] (6.0, \yw) -- (6.45, \yw);
\node[lbl, above] at (6.25, \yw) {$\hat\sigma_3$};


\draw[thick] (-1.8, -\yw) -- (I0.west |- 0,-\yw);
\draw[arr] (-1.5, -\yw) -- (-1.05, -\yw);
\node[lbl, below] at (-1.25, -\yw) {$\hat\varrho_0$};

\draw[thick] (I0.east |- 0,-\yw) -- (I1.west |- 0,-\yw);
\draw[arr] (1.0, -\yw) -- (1.45, -\yw);
\node[lbl, below] at (1.25, -\yw) {$\hat\varrho_1$};

\draw[thick] (I1.east |- 0,-\yw) -- (I2.west |- 0,-\yw);
\draw[arr] (3.5, -\yw) -- (3.95, -\yw);
\node[lbl, below] at (3.75, -\yw) {$\hat\varrho_2$};

\draw[thick] (I2.east |- 0,-\yw) -- (6.8, -\yw);
\draw[arr] (6.0, -\yw) -- (6.45, -\yw);
\node[lbl, below] at (6.25, -\yw) {$\hat\varrho_3$};

\node[font=\Huge] at (7.2, 0) {$\cdots$};

\end{tikzpicture}%
}
\caption{\label{fig:chain} A realization of \cref{eq:Kn} wherein a
``system" and a ``bath" evolve jointly under a Markovian map with negligible entanglement.}
\end{figure}
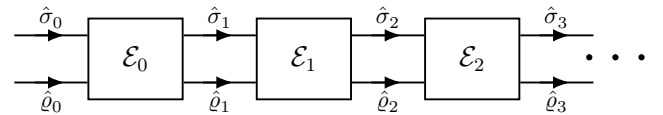

But first, we clarify how such maps can be realized in a model of system-bath interactions.
Consider the discrete-time process sketched in \cref{fig:chain}:
at each time $t_n$, the full description of the process is given by a density matrix 
$\hat \rho_n \otimes \hat \sigma_n$, where $\hat \rho_n$ is the state of a ``system" (A) and $\hat \sigma_n$ is the state of a ``bath" (B).
Assume that the system and bath evolve under a joint Markovian evolution that results in 
negligible entanglement at each step. 
Then, at time $t_{n+1}$ the state of the full system is
\begin{equation}\label{eq:JointMarkEvo}
	\hat \rho_{n+1} \otimes \hat \sigma_{n+1} = \mathcal{E}_n[\hat \rho_n \otimes \hat \sigma_n],
\end{equation}
where $\mathcal{E}_n$ is a (single-variable) completely positive trace-preserving map.
The reduced states of the system and bath are respectively,
$\hat \rho_{n+1} = \Tr_B \mathcal{E}_n[\hat \rho_n \otimes \hat \sigma_n]$ and $\hat \sigma_{n+1} = \Tr_A \mathcal{E}_n[\hat \rho_n \otimes \hat \sigma_n]$.
Recursively eliminating the bath state $\hat\sigma_n$ between these pair of equations leads
to an expression for $\hat \rho_{n+1}$ in terms of all previous system states $\hat \rho_n, \ldots, \hat \rho_0$, i.e. a map of the form in \cref{eq:Kn}.
For \cref{eq:JointMarkEvo} to hold, each $\mathcal{E}_n$ must be an entanglement-breaking 
channel, implying that the bath effectively carries classical information \cite{HorShoRus03}.
In the sequel we directly work with maps of the form in \cref{eq:Kn}, without reference to 
any realization.


\emph{Axioms for the maps $\mathcal{K}_n$.}
We are interested in the maps $\{\K_n\}$ in \cref{eq:Kn}, each a multi-variable map 
acting on $n$ states, that satisfy some natural axioms. 

A natural demand is that the maps $\K_n$ are compatible with the system being part of a larger Hilbert space, which is axiomatized by the notion of complete positivity. 
In the familiar case of a single-variable map $\K_1$, it is completely positive (CP) 
if $\K_1 \otimes \mathcal{I}_m$ 
is positive for all $m$, where $\mathcal{I}_m$ is the identity map on an $m$-dimensional Hilbert space. 
The extension to multi-variable maps is similar \cite{AndChoi86}: 
a map of $n+1$ variables $\K_{n+1}[\hat \rho_n, \ldots, \hat \rho_0]$ is 
CP iff., for any set of $n+1$ positive block-matrices $\hat\rho_m$ (with $(i,j)$'th block 
$\hat \rho_{m,ij}$, $0 \leq m \leq n$), the block-matrix with $(i,j)$'th block $\K_{n+1}[\hat \rho_{n,ij}, \ldots, \hat \rho_{0,ij}]$ is also positive.
It captures the essential requirement that the maps $\K_n$ produce positive states even when
each of their arguments is possibly entangled with a spectator ancilla.

Ando and Choi \cite{AndChoi86} have shown that any multi-variable CP map
can be represented as a sum of multi-variable CP maps that are homogeneous in each of their arguments.
(or anti-homogeneous, which we ignore in the context of maps between quantum states).
That is, if $\K$ is a multi-variable CP map, then there exists a decomposition
\begin{equation}\label{eq:CPdecomposition}
	\K[\hat \rho_n,\ldots, \hat \rho_0] = \sum_{m_1, \ldots, m_n} \K_{m_1, \ldots, m_n}[\hat \rho_n,\ldots, \hat \rho_0],
\end{equation}
where $m_i$ are positive integers and $\K_{m_1, \ldots, m_n}$ is homogeneous
of degree $m_i$ in the $i$'th argument, i.e. 
\begin{equation}\label{eq:CPhomogeneity}
\begin{split}
	&\K_{m_1, \ldots, m_n}[\hat \rho_n,\ldots, z \hat \rho_i, \ldots, \hat \rho_0]\\ 
	&\qquad\qquad = z^{m_i} \K_{m_1, \ldots, m_n}[\hat \rho_n,\ldots, \hat \rho_i, \ldots, \hat \rho_0].
\end{split}
\end{equation}
for any $z \in \mathbb{C}$. 
This is the structural implication of complete positivity.

In addition to complete-positivity, we demand that the maps $\K_n$ are non-signaling (NS).
The single-variable map $\K_1$ is NS \cite{Ghir88,Gisin89,SimGis01} iff.
it is convex-linear, i.e. 
$\mathcal{K}_1[\sum_\alpha p_\alpha \hat{\rho}_\alpha] = \sum_\alpha p_\alpha 
\mathcal{K}_1[\hat{\rho}_\alpha]$, for any states $\hat{\rho}_\alpha$ and any probabilities 
$\{p_\alpha\}$. 
Convex-linearity ensures that statistically equivalent ensembles cannot 
be distinguished by the map $\K_1$, so no subsequent local measurement 
can reveal which ensemble was remotely prepared; this is precisely non-signaling.
We can 
extend the notion of non-signaling to multi-variable maps by demanding that they are convex-linear
in each of their arguments, i.e.
\begin{equation}\label{eq:ConvexLinearity}
\begin{split}
	\K_{n+1}&\left[\sum_\alpha p_\alpha \hat{\rho}_{n,\alpha}, \cdots, 
	\sum_\alpha p_\alpha \hat{\rho}_{0,\alpha}\right]\\ 
	&\qquad \qquad\qquad= \sum_\alpha p_\alpha \K_{n+1}[\hat{\rho}_{n,\alpha}, \ldots, \hat{\rho}_{0,\alpha}],
\end{split}
\end{equation}
for any states $\hat{\rho}_{m,\alpha}$ and any probabilities $\{p_\alpha\}$. 

Convex-linearity can be conveniently reformulated in terms of sub-normalized states.
Note that the single-variable map $\K_1$ is defined on the set of normalized states.
The normalized states form the base of a cone of sub-normalized
states $\{p \hat{\rho} : p \in [0,1], \Tr \hat{\rho}=1\}$.
Any convex-linear map on the base of a convex cone has a unique extension 
to the whole cone which is homogeneous of degree one (see End Matter).
That is, $\K_1$ has a unique extension to sub-normalized states: 
$\K_1[p\hat{\rho}] = p \K_1[\hat{\rho}]$ for any $p\in[0,1]$. 
It can be shown that convex-linear multi-variable maps have a similar extension, i.e.
\begin{equation}\label{eq:TraceCond}
	\K_{n+1}[p \hat \rho_n, .., p \hat \rho_0] = p \, \K_{n+1}[\hat \rho_n, .., \hat \rho_0],
\end{equation}
for any $p\in [0,1]$ and any normalized states $\{\hat\rho_n\}$ (see End Matter).
This is an equivalent formulation of non-signaling and has the interpretation that
if the state $\hat \rho_0$ is initialized with probability $p$, represented by the sub-normalized
state $p \hat \rho_0$, then the consequent evolution will happen with the same probability $p$.

\emph{Structure of CPNS maps.}
We now identify the most general form of completely-positive non-signaling (CPNS) maps,
i.e. those that satisfy \cref{eq:CPdecomposition,eq:CPhomogeneity,eq:TraceCond}.
Note that \cref{eq:TraceCond} constrains the values of $m_i$ that can appear in the expansion
in \cref{eq:CPdecomposition}:
only the terms with $\sum_i m_i = 1$ satisfy \cref{eq:TraceCond} for all $0 \leq p \leq 1$, so 
$m_{i} = \delta_{ik} $ for some $0 \leq k \leq n$; i.e. any multi-variable CP 
map $\K[\hat \rho_n, \ldots, \hat \rho_1]$ is a sum of one-argument maps $\K_{0,.., m_k = 1, ..,0}[\hat \rho_k]$. 
But since each of such CP maps is also linear, according to \cite{AndChoi86}, it has a Kraus representation. 
Thus, an arbitrary CPNS map $\K_{n+1}$ can be expressed in the form
\begin{equation}\label{eq:NonMarkKraus}
    \mathcal{K}_{n+1}[\hat \rho_n, .., \hat \rho_0] = 
	\sum_{m=0}^n \sum_\alpha V^\dagger_{\alpha n m} \hat \rho_m V_{\alpha n m},
\end{equation}
where the operators $V_{\alpha nm}$ satisfy
$\sum_\alpha V_{\alpha nm} V^\dagger_{\alpha nm} = p_{nm} \hat I$ with 
$p_{nm} \geq 0$ and $\sum_{m=0}^n p_{nm} = 1$, 
to ensure that $\mathcal{K}_{n+1}$ is trace-preserving.
This structure subsumes the Kraus representation of a single-variable ($n=0$) 
CP map.

\emph{Continuous-time limit.}
To elucidate the continuous-time dynamics corresponding to the discrete-time CPNS maps in 
\cref{eq:NonMarkKraus}, we suppose that the evolution is continuous in the sense that 
$\hat{\rho}_{n+1} =\hat{\rho}_n + O(\Delta t)$, 
where $\Delta t = t_{n+1} - t_n$, and take the limit $\Delta t \to 0$.

Continuity of the evolution constrains the scaling of the operators $V_{\alpha nm}$ 
with $\Delta t$: the required continuous behavior can only be produced if
\begin{equation}\label{eq:Vscaling}
\begin{split}
	V_{\alpha nm} = 
	\begin{dcases}
	  1 + \Delta t\, {L}_0(t_n) + O(\Delta t^2); & \alpha = 0, m = n\\
	  \sqrt{\Delta t}\,  L_\alpha(t_n) + O(\Delta t); & \alpha \neq 0, m = n\\
	  \Delta t\, {R}_\alpha(t_n, t_m) + O(\Delta t^2); & \alpha \neq 0, m < n,
    \end{dcases}
\end{split}
\end{equation}
where ${L}_\alpha(t), {R}_\alpha(t,t')$ are operators independent of $\Delta t$,
and $\sum_\alpha R_\alpha(t, t') R^\dagger_\alpha(t, t') = p(t, t') \hat I$ where $p(t, t') \geq 0$
is the rate density for invoking the state $\hat \rho(t'<t)$ in the evolution of $\hat \rho(t)$.

The terms in \cref{eq:Vscaling} represent qualitatively different contributions 
to $\hat{\rho}_{n+1}$.
The terms with $m=n$ represent the action of Kraus operators on the
current state $\hat{\rho}_n$, and are thus the usual Markovian contribution to the dynamics
exhibiting their well-known scaling with $\Delta t$ \cite{Manzano2020}.
The terms with $m < n$ represent the action of Kraus operators on all past states 
$\{\hat{\rho}_m\}_{m < n}$, and are thus the non-Markovian contribution. 
The scaling with $\Delta t$ here leads to the
Kraus sum scaling as $\sum_m^n (\Delta t)^2$, but since there are $n \sim 1/\Delta t$ terms in the sum, 
the total contribution of this non-Markovian part scales as
$\Delta t$, ensuring a well-defined continuous-time limit.   

Substituting the parametrization for $V_{\alpha nm}$ into \cref{eq:NonMarkKraus} and taking the limit $\Delta t \to 0$ leads to an integro-differential equation for $\hat\rho(t)$:
\begin{multline}\label{eq:EqwithL0}
	\frac{\p \hat \rho}{\p t} = 
	L_0^\dagger(t) \hat \rho(t) + \hat \rho(t) L_0(t) 
	+ \sum_{\alpha \neq 0} L_\alpha^\dagger (t) \hat \rho(t) L_\alpha(t)\\
	+ \int_{t_0}^t \dd{t'} \sum_\alpha R_\alpha^\dagger (t, t') \hat \rho(t') R_\alpha (t, t').
\end{multline}
Splitting $L_0(t) = G(t) + iH(t)$ into Hermitian $G(t)$ and $H(t)$ and imposing $\Tr[\p_t \hat \rho(t)] = 0$ fixes $G(t)$. 
After the explicit separation of the Markovian and non-Markovian contributions, the non-Markovian master equation attains its final form:
\begin{equation}\label{eq:MasterEquation}
	\frac{\p \hat \rho}{\p t} = (\mathcal{L}\hat \rho)(t) 
	+ \int_{t_0}^{t} \dd{t^\prime} (\mathcal{R}\hat \rho)(t, t^\prime).
\end{equation}
The generator for Markovian part of evolution has a form $(\mathcal{L}\hat \rho)(t) 
= - i [H(t), \hat \rho(t)] + \sum_{\alpha}\!\left( L^\dagger_\alpha \hat \rho L_\alpha - \tfrac{1}{2}\{ L_\alpha L_\alpha^\dagger; \hat \rho \}\right)$, which is the familiar GKSL generator. 
The operators $L_\alpha = L_\alpha(t)$ may depend explicitly on the current time. 
The non-Markovian part of the evolution is governed by
\begin{equation}\label{eq:RExpression}
	(\mathcal{R}\hat \rho)(t, t^\prime) = \sum_\alpha \left( R_\alpha^\dagger\hat\rho(t^\prime) R_\alpha - \frac{1}{2} \{ R_\alpha R_\alpha^\dagger; \hat \rho(t) \} \right),
\end{equation}
where $R_\alpha = R_\alpha(t, t')$ are allowed to depend on both current and retarded times.
The structure of \cref{eq:RExpression} looks like that of a GKSL generator, but the jump term
acts on the retarded density matrix $\hat \rho(t^\prime)$ and the anti-commutator 
on the current state $\hat \rho(t)$; 
this asymmetry is not incidental, but is a direct consequence of complete positivity 
and non-signalling. (See Sec. I of Supplemental Information \cite{Supplement} for details.)
This structure is subtly different from other proposals for non-Markovian 
dynamics \cite{Nakajima,Zwanzig,Shabani2005,Vacch16}, but it is this difference that guarantees 
that \cref{eq:MasterEquation} represents a CPNS evolution.

The construction above is valid whenever $R_\alpha(t,t')$ are bounded and $L_\alpha(t)$ 
are closed \cite{Lind76,Dav79,DunfordSchwartz}.
The converse is true in the sense that any evolution of the form in \cref{eq:MasterEquation}
can be expressed as \cref{eq:NonMarkKraus}, whenever $L_\alpha(t)$ and $R_\alpha(t,t')$ are
bounded and $p(t,t')$ is integrable on the interval $[0, t]$.
However, for unbounded $L_\alpha(t)$, there is no general result, similar in spirit 
to \cite{Dav79}, that would establish the converse.

\emph{Continuous-time dynamics of observables.}
The non-Markovian master equation 
\Cref{eq:MasterEquation} defines a linear map $\hat{\rho}(t_0) \mapsto \hat{\rho}(t) = 
\mathcal{K}_{t,t_0}\hat{\rho}(t_0)$ on the set of states, where $\mathcal{K}_{t,t_0}$ is the
continuous-time version of the map $\K_n$ in \cref{eq:Kn}.
The non-Markovian nature of the evolution implies that $\mathcal{K}_{t,t_0}$ does not satisfy the composition law, i.e.
$\mathcal{K}_{t_2,t_0} \neq \mathcal{K}_{t_2,t_1}\circ \mathcal{K}_{t_1,t_0}$ 
for $t_2 > t_1 > t_0$ in general.
This means that, although the adjoint map $\mathcal{K}_{t,t_0}^*$ formally evolves observables by $\hat{X}(t_0)\mapsto \hat X(t) = \mathcal{K}_{t,t_0}^* \hat X(t_0)$, 
the equations of motion for observables do not close, and there is no regression theorem for correlation functions of observables.
Nevertheless, correlation functions can be defined
and evaluated by specifying measurement interventions and propagating the
measurement-conditioned state in the intervening duration using \cref{eq:MasterEquation}.

\emph{Instantaneous measurements.}
We first consider the case of an instantaneous measurement defined by
the state transformation $\hat \rho(t) \mapsto 
(\mathcal J_\alpha \hat \rho)(t) = J_\alpha^\dagger \hat \rho(t) J_\alpha $, 
where $J_\alpha$ is an effect operator corresponding to the measurement outcome $\alpha$.
This measurements defines a classical random variable $I_\alpha(t) \in \{1,0\}$, 
depending on whether $\alpha$ is observed or not.
Its classical expectation is given by $\EE[I_\alpha(t)]= \Tr\!\left[(\J_\alpha\hat{\rho})(t)\right]$.
Its correlation function $\EE[I_\alpha(t+ \tau)I_\beta(t)]$ corresponds to the action of 
the corresponding operator at each of the two times. 
Since $I_\alpha(t)$ and $I_\beta(t)$ are valued in $\{1,0\}$, 
the probability $P[I_\beta(t) = 1] = \EE[I_\beta(t)]$.
Then
\begin{equation*}
	\EE[I_\alpha(t + \tau) I_\beta(t)] = \EE[I_\alpha(t+\tau) \mid I_\beta(t) = 1] 
	P[I_\beta(t) = 1],
\end{equation*}
where the conditional expectation value is given by
\begin{equation}\label{eq:CondExpect}
	\EE[I_\alpha(t+\tau) \mid I_\beta(t) = 1] = \Tr[\mathcal J_\alpha 
	\hat \rho(t+\tau\mid I_\beta(t) = 1)],
\end{equation}
with $\hat \rho(t+\tau| I_\beta(t) = 1)$ the state of the system at the time $t+\tau$ 
conditioned on $I_\beta(t) = 1$ at the time $t$.
The event $I_\beta(t) = 1$ implies that $\mathcal J_\beta$ is applied at time $t$ to 
the state $\hat \rho(t)$; thus the conditional state
\begin{equation}\label{eq:CondState}
	\hat \rho(t' | I_\beta(t) = 1) = 
\begin{cases}
	(\mathcal J_\beta \hat \rho)(t)/P[I_\beta(t) = 1], & t' = t,\\
	\hat \rho(t'), & t' < t.
\end{cases}
\end{equation}
By causality, measurements performed in the future cannot retroactively affect the current state 
of the system, so the state at $t' < t$ remains unchanged.
The conditional state $\hat \rho(t+\tau| I_\beta(t) = 1)$ is obtained by evolving the system in the state $\hat \rho(t| I_\beta(t) = 1)$ from time $t$ to $t+\tau$ by
\begin{equation}\label{eq:DiffCorrEq}
	\p_\tau \hat \rho_{t}(t+\tau) = (\mathcal{L}\hat\rho_{t})(t+\tau)
	 + \int_{t_0}^{t+\tau} \dd{t^\prime} (\mathcal{R}\hat \rho_{t})(t+\tau, t^\prime),
\end{equation}
where $\hat \rho_t(t') = \hat \rho(t'| I_\beta (t) = 1)$ and the initial conditions at time 
$t$ are given by \cref{eq:CondState}.
This state then allows to evaluate \cref{eq:CondExpect} and 
therefore $\EE[I_\beta(t+\tau)I_\alpha(t)]$.

A more general non-binary random variable can be thought of as a weighted sum of binary-valued
random variables $I(t) = \sum_\alpha \nu_\alpha I_\alpha(t)$, which corresponds to a measurement defined by the operator $\mathcal J = \sum_\alpha \nu_\alpha \mathcal J_\alpha$.
Clearly, its expectation value is $\EE[I(t)] = \Tr[(\mathcal J \hat \rho)(t)]$,
and the correlation function is given by
\begin{equation}\label{eq:ClassicSum}
	\EE[I(t+ \tau)I(t)] = \sum_{\alpha \beta} \nu_\alpha \nu_\beta \EE[I_\alpha(t+ \tau) I_\beta(t)].
\end{equation}
The bilinearity of \cref{eq:ClassicSum} in weights $\nu_\alpha$ and the linearity of \cref{eq:DiffCorrEq} allows to define a conditional state $\hat \rho_{J, t}(t')$  
to compute $\EE[I(t+\tau)I(t)]$ directly:
\begin{equation}\label{eq:EII}
	\EE[I(t+\tau)I(t)] = \Tr[(\mathcal J \hat \rho_{J, t})(t+\tau)], \quad \tau > 0.
\end{equation}
The state $\hat \rho_{J, t}(t+\tau)$ is a result of evolution according to \cref{eq:DiffCorrEq} from $t$ to $t+\tau$, where $\hat \rho_t(t+\tau) = \hat \rho_{J, t}(t+ \tau)$ and the initial conditions at time $t$ are given by 
\begin{equation}\label{eq:RhoJ}
	\hat \rho_{J, t}(t') = 
\begin{cases}
	(\mathcal J \hat \rho)(t), &   t' = t,
	\\
	\EE[I(t)] \hat \rho(t'), & t' < t.
\end{cases}
\end{equation}

A particularly useful special case is a weak measurement of an observable $\hat{X}$.
It is obtained by specializing to 
${(\J\hat\rho)(t) = (x + \hat{X})^\dagger\hat\rho(t)(x+\hat{X})}$, which, to
leading order in $x$ is,
${\Tr[(\J\hat \rho)(t)] \approx x^2 + 2 x \Re\!\t{Tr}[\hat{X} \hat \rho(t)]}$.
We interpret this as defining a current of large mean value $x^2$, with fluctuations 
described by $\hat{X}$.
Subtracting the mean and rescaling defines the diffusive current
${I_X(t) \equiv x^{-1}(I(t)-x^2)}$.
The two-time correlation of this current is given, to leading order in $x$, by
\begin{equation}\label{eq:IXCorr}
	\EE[I_X(t+\tau)I_X(t)] = \Tr[(\mathcal X \hat\rho_{X,t})(t+\tau)], \quad \tau > 0,
\end{equation}
where we define $( \mathcal X \hat\rho)(t) \equiv \hat{X}^\dagger \hat\rho(t) + \hat\rho(t) \hat{X}$.
The quantity $\hat\rho_{X,t}(t+ \tau)$ is defined by evolving $\hat\rho_{X,t}(t)$ from time $t$ to $t+ \tau$ according to \cref{eq:DiffCorrEq}, where $\hat \rho_{t}(t+ \tau) = \hat \rho_{X, t}(t+ \tau)$, and the initial conditions are given by 
\begin{equation}\label{eq:RhoX}
	\hat \rho_{X, t}(t') = 
\begin{cases}
	(\mathcal X \hat \rho)(t), &   t' = t,
	\\
	2\Re\t{Tr}[\hat{X} \hat \rho(t)] \hat \rho(t'), & t' < t.
\end{cases}
\end{equation}

\emph{Continuous-time measurement.}
So far, the measurement intervention is assumed to be instantaneous. 
However, the scope of \cref{eq:EII} and \cref{eq:DiffCorrEq} can be expanded to 
time-continuous measurements using the method described in Ref. \cite{Landi2024}.
A time-continuous measurement can be modeled by incorporating it into the continuous dynamics of the
system: the measurement at each time is described by $(\mathcal J_\alpha^{(c)} \hat \rho)(t) \dd{t} 
= L_\alpha^\dagger \hat \rho L_\alpha (t) \dd{t}$, where $L_\alpha$ are some Lindblad operators 
that appear in \cref{eq:MasterEquation}. 
The back-action of the measurement is automatically included in the dynamics.
As before, a weighted current $I^{(c)}(t) = \sum_\alpha \nu_\alpha I_\alpha^{(c)}(t)$ can be defined, 
corresponding to the operation $\mathcal J = \sum_\alpha \nu_\alpha \mathcal J_\alpha$. 
It has the expectation value $\EE[I^{(c)}(t)] = \t{Tr}[(\mathcal J^{(c)} \hat \rho)(t)] $. 
Its correlation function can be expressed as
\begin{equation}\label{eq:ContMeasurement}
	\EE[I^{(c)}(t+\tau)I^{(c)}(t)] = \Tr[(\mathcal J \hat \rho_{J, t}^{(c)})(t+\tau)] + K \delta(t),
\end{equation}
where $K = \sum_\alpha \t{Tr}[\nu_\alpha^2 L_\alpha^\dagger \hat \rho L_\alpha]$.
The additional contribution $K \delta(t)$ is the $O(dt^{-1})$ contribution that arises 
from $\tau = 0$, where two measurements overlap.
The conditional state $\hat \rho_{J, t}^{(c)}(t+ \tau)$ is the state evolved from $t$ to $t+ \tau$ by \cref{eq:DiffCorrEq} with the same initial conditions at time $t$ as in \cref{eq:RhoJ}, but with $\mathcal J^{(c)}$ and $I^{(c)}$ being substituted for $\mathcal J$ and $I$ respectively.
Similarly, continuous-time weak measurement of the observable $\hat X$ can be described by the corresponding current with boundary conditions identical to \cref{eq:IXCorr,eq:RhoX}.

\emph{Application to a two-level system.}
We illustrate the above formalism on a driven two-level system (TLS) coupled to a bosonic bath. 
We model it by the Hamiltonian $H(t) = \omega_0 \sigma^\dagger \sigma 
+ \Omega \sigma e^{i \nu t} + \mathrm{h.c.}$, where $\sigma$ is the annihilation operator of 
the TLS, $\omega_0$ is the splitting between its two levels, and $\Omega e^{i \nu t}$ is a time-dependent drive.
Its interaction with a bosonic bath is described by the interaction Hamiltonian 
$H_I = g(\sigma^\dagger E + \sigma E^\dagger)$, where $E = \int_0^\infty \dd{\omega}
f_E(\omega)b(\omega)$ with $b(\omega)$ the annihilation operator of the bath mode 
distributed according to $f_E$, and $g$ is the coupling strength.
We assume that the bath is stationary, i.e. $\Tr_\t{B} [b^\dagger(\omega) b(\omega^\prime) \hat \rho_\mathrm{B}] = n(\omega) \delta(\omega - \omega^\prime)$, where $n(\omega)$ is the average 
occupation number of the bath.
A standard textbook exercise in the weak-coupling limit \cite{BP02} of applying the Born approximation, but crucially, not the Markov approximation, leads directly 
to \cref{eq:MasterEquation} where the non-Markovian contribution is given by 
\cref{eq:RExpression} with ${R_1(t, t') 
= S_E^{1/2}(t-t')\sigma^\dagger}$ and ${R_2(t,t') = S_E^{1/2}(t-t')\sigma}$,
where $S_E(t) = 2 g^2 \int_0^\infty d\omega\,|f_E(\omega)|^2 n(\omega)\cos((\omega-\omega_0)t)$
is the correlation function of the bath operator.
Bochner's theorem guarantees $S_E(t) \geq 0$ for all times, ensuring 
that $\sum_{i} R_i^\dagger(t, t') R_i(t, t') = S_E(t - t') \hat I$ is 
positive semi-definite, guaranteeing that the dynamics is CP.
The Markovian limit is recovered when $|f_E(\omega)|^2 n(\omega) = \mathrm{const}$, i.e. white, 
and $S_E(t) \propto \delta(t)$.
(See Sec. II of Supplemental Information \cite{Supplement} for details.)
We further include Markovian dissipation by adding Lindblad operators 
$L_1 = \sqrt{\Gamma_\mathrm{M}} \sigma^\dagger$ and $L_2 = \sqrt{\bar \Gamma_\mathrm{M}} \sigma$ 
to the master equation.
The resulting form of  \cref{eq:MasterEquation} is a Volterra integro-differential equation of 
the second kind, which can be solved (see Sec. III of Supplemental Information \cite{Supplement}).

We consider the continuous weak measurement of the operator $\sigma$, which defines the 
emission spectrum $S(\omega) = \int_0^\infty \cos(\omega \tau) \EE [ I_\sigma(t+ \tau) 
I_\sigma (t)]$.
Using the formalism developed above, and the solution of the master equation, we can evaluate
the emission spectrum.
In the limit of small detuning and large pumping rate, it takes the form
(see Sec. III of Supplemental Information \cite{Supplement} for details)
\begin{equation}\label{eq:NonMarkMollow}
    S(\omega + \omega_0) = \mathrm{Re}\!\left[\frac{1}{4\lambda_0} + \frac{1}{8\lambda_+} + \frac{1}{8\lambda_-}\right],
\end{equation}
where 
$\lambda_0 \approx -i\omega + \Gamma_\mathrm{T}$, 
and 
${\lambda_\pm \approx -i\omega \pm 2i\Omega  + 3\gamma^\prime_\mathrm{T}(\omega)/2}$, 
with $\gamma^\prime_\mathrm{T}(\omega) = (\Gamma_\mathrm{M} + \bar\Gamma_\mathrm{M} + \frac{4}{3} \Gamma_\mathrm{NM} +  \frac{2}{3}\gamma_\mathrm{NM}(\omega))/2$ and ${\Gamma_\mathrm{T} = \gamma^\prime_\mathrm{T}(0)}$. 
This is the non-Markovian analog of the Mollow triplet \cite{Mollow1969}: the three-peak 
structure of the spectrum is preserved, but each peak acquires a frequency-dependent 
linewidth governed by $\gamma_{\mathrm{NM}}(\omega) = \int_0^{\infty} S_E(t) e^{i \omega t} \dd{t}$, 
directly encoding the memory of the non-Markovian bath.
This is qualitatively similar to the phenomenological description of non-Markovian damping of a
harmonic oscillator by a frequency-dependent damping rate \cite{Saul90}, but our formalism
not only derives it from first principles, but also enables assigning a legitimate notion of a
quantum state to the system at all times, in contrast to the phenomenological approach.


\emph{Conclusion.}
We have defined a novel class of discrete-time non-Markovian quantum stochastic processes on 
the space of quantum states, and characterized its form under the assumption of complete positivity
and non-signalling.
Its continuous-time limit, when it exists, is given by a non-Markovian generalization of the
GKSL equation given by \cref{eq:MasterEquation}.
This framework is directly applicable to open quantum systems driven by noises with 
arbitrary integrable power spectral density, and allows a consistent definition of the 
quantum state of such a system at all times.
We have shown how this state determines the results of continuous measurements of the system, 
circumventing the failure of the regression theorem in the non-Markovian regime, and allowing all
operationally meaningful predictions to be extracted from the theory. 
All of this is illustrated on the example of a driven two-level system coupled to a bosonic
bath without the Markov approximation; the resulting non-Markovian analog of the Mollow triplet
encodes the memory of the bath in the frequency-dependent linewidth of each peak.

Open problems include the extension of the formalism to deal with noises with 
non-integrable power spectral density (such as $1/f$ noise), which will enable the consistent
assignment of a quantum state to such systems at all times; and a full quantum trajectory 
description of such a system, which will enable full state estimation and feedback control of 
non-Markovian systems.

\emph{Acknowledgment.} This research is funded
in part by an NSF CAREER award (PHY–2441238)
and by the Gordon and Betty Moore Foundation (grant
GBMF13780).


\bibliography{refs_nonMarkovian}

\clearpage
\appendix
\onecolumngrid

\begin{center}
{\bf\large End Matter}
\end{center}

\section*{Extension of an affine map on the base of a convex cone to the whole cone}

The requisite characterization, namely \cref{thm:ConeExtension} below, can be extracted from standard
results at the intersection of convex analysis and functional analysis \cite{AsiEll80}. 
Here we give a self-contained elementary proof.
The following definitions, required to state the result, are standard in convex analysis \cite{Rock70}.

A \emph{cone} $C$ is a subset of a real Banach space $V$ such that if $x \in C$, 
then $\alpha x \in C$ for all $\alpha > 0$.
The cone has a \emph{base} iff. there exists a continuous linear functional $f: V \to \mathbb{R}$ 
such that $f(x) > 0$ for all $x \in C$; then the base of the cone with respect to $f$ is the set
$B_f(C) = \{ x \in C \mid f(x) = 1 \}$, which is a subset of $C$.

A \emph{convex set} $S$ is a set
such that if $x, y \in S$, then $\lambda x + (1-\lambda) y \in S$ for all $\lambda \in [0,1]$. 

A \emph{convex cone} is a cone that is also a convex set. 
Let $C$ be a convex cone. Then for any $x,y \in C$ and $\alpha,\beta > 0$, we have $\alpha x + \beta y \in C$ since
$\alpha x + \beta y = (\alpha +\beta)[\tfrac{\alpha}{\alpha+\beta}x + \tfrac{\beta}{\alpha+\beta}y]$,
which is $(\alpha +\beta)$, a positive number, 
times a convex combination of $x$ and $y$.
Thus, a convex cone is closed under addition and positive scalar multiplication.

The base of a convex cone is a convex set. Suppose $b,b' \in B_f(C)$ and $\lambda \in [0,1]$. Then $\lambda b + (1-\lambda) b' \in C$ by the convexity of the cone; furthermore, since $f$ is linear, 
$f(\lambda b + (1-\lambda) b') = \lambda f(b) + (1-\lambda) f(b') = \lambda\cdot 1 + (1-\lambda)\cdot 1 = 1$; so $\lambda b + (1-\lambda) b' \in B_f(C)$.

\begin{lemma}\label{lem:UniqueDecomp}
	Every nonzero $x\in C$ can be uniquely expressed as $x = \alpha b$ for some $\alpha > 0$ and $b \in B_f(C)$.
\end{lemma}
\begin{proof}
	Take $\alpha = f(x)$ and $b = x/f(x)$. 
	Since $f(b) = f(x/f(x)) = f(x)/f(x)=1$, we have $b \in B_f(C)$. 
\end{proof}

An \emph{affine} (or convex-linear) map $h$ between convex sets is a map such that for all $x,y$ in the domain and $\lambda \in [0,1]$, 
$h(\lambda x + (1-\lambda) y) = \lambda h(x) + (1-\lambda) h(y)$.

\begin{theorem}\label{thm:ConeExtension}
Let $C$ be a convex cone with a base $B_f(C)$, and let $h: B_f(C) \rightarrow B_f(C)$ be a bounded affine
map. Then there exists a unique bounded linear map $\tilde h$ on the cone such that 
(i) $\tilde h(\alpha x) = \alpha \tilde{h}(x)$ for all $x\in C$ and $\alpha > 0$; 
(ii) $\tilde h\vert_{B_f(C)} = h$; and 
(iii) $f(\tilde h(x)) = f(x)$ for all $x \in C$. 
\end{theorem}

\begin{proof}
Define $\tilde h: C \to C$ by 
\begin{equation*}
	\tilde h(x) = f(x)\, h\!\left(\frac{x}{f(x)}\right).
\end{equation*}
This is well-defined: $f(x) > 0$ implies $x/f(x) \in B_f(C)$, so $h(x/f(x)) \in B_f(C) \subset C$.

\noindent The boundedness of $\tilde h$ follows from $\| \tilde h (x) \| = |f(x)| \| h(x/f(x))\| \leq C_h \| x\|$, where $C_h$ is the norm of $h$.

\noindent The conditions (i), (ii), and (iii) are easy to verify:

(i) For $\alpha > 0$: $\tilde h(\alpha x) = f(\alpha x)\, h\!\left(\frac{\alpha x}{f(\alpha x)}\right) = \alpha f(x)\, h\!\left(\frac{x}{f(x)}\right) = \alpha\, \tilde h(x)$,
using the linearity of $f$.

(ii) For $b \in B_f(C)$: $\tilde h(b) = 1 \cdot h(b) = h(b)$.

(iii) $f(\tilde h(x)) = f\!\left(f(x)\, h\!\left(\frac{x}{f(x)}\right)\right) = f(x)\, f\!\left(h\!\left(\frac{x}{f(x)}\right)\right) = f(x)$,
since $h(x/f(x)) \in B_f(C)$ implies $f(h(x/f(x))) = 1$.

\noindent \emph{Linearity.} We first show $\tilde h(x+y) = \tilde h(x) + \tilde h(y)$ for $x,y \in C$.
Set $\alpha = f(x)$, $\beta = f(y)$, so that $x/\alpha,\, y/\beta \in B_f(C)$.
Since $C$ is a convex cone, $x+y \in C$, and by linearity of $f$, $f(x+y) = \alpha + \beta$; thus 
by the definition of $\tilde h$:
\begin{equation*}
\tilde h(x+y) = (\alpha + \beta)\, h\!\left(\frac{x+y}{\alpha + \beta}\right).
\end{equation*}
The last factor can be simplified by writing
\begin{equation*}
	\frac{x+y}{\alpha + \beta} = \underbrace{\frac{\alpha}{\alpha+\beta}}_{\lambda}\frac{x}{\alpha} 
	+ \underbrace{\frac{\beta}{\alpha+\beta}}_{1-\lambda}\frac{y}{\beta}
\end{equation*}
and using the affineness of $h$:
\begin{equation*}
\begin{split}
	\tilde h(x+y) 
	= (\alpha+\beta)\left[\frac{\alpha}{\alpha+\beta}\, h\!\left(\frac{x}{\alpha}\right) 
	+ \frac{\beta}{\alpha+\beta}\, h\!\left(\frac{y}{\beta}\right)\right]
	= \alpha\, h\!\left(\frac{x}{\alpha}\right) + \beta\, h\!\left(\frac{y}{\beta}\right) 
	= \tilde h(x) + \tilde h(y).
\end{split}
\end{equation*}
Combined with (i), this gives $\tilde h(\alpha x + \beta y) = \alpha\,\tilde h(x) + \beta\,\tilde h(y)$ for all $\alpha,\beta > 0$, so $\tilde h$ is linear on $C$.

\noindent \emph{Uniqueness} follows from \cref{lem:UniqueDecomp} and (i), (ii):
let $h': C \rightarrow C$ be any linear map satisfying (i), (ii), then
\begin{equation*}
	h'(x) = h'\!\left(f(x)\cdot\frac{x}{f(x)}\right) 
	\overset{(\mathrm{i})}{=} f(x)\, h'\!\left(\frac{x}{f(x)}\right) 
	\overset{(\mathrm{ii})}{=} f(x)\, h\!\left(\frac{x}{f(x)}\right) = \tilde h(x).
\end{equation*}
\end{proof}

The set of normalized quantum states forms the base of the convex cone of
sub-normalized states, with the trace $\t{Tr}$ identified as the linear functional $f$ that defines the base.
With this identification, \cref{thm:ConeExtension} immediately applies to the single-variable map $\K_1$: 
if $\K_1$ is convex-linear on the set of normalized states, it has a unique
extension, which we also denote $\K_1$, to a linear map on the sub-normalized states such that $\K_1[\lambda \rho] = \lambda \K_1[\rho]$.

For the multi-variable map $\K_{n+1}$, convex-linearity is postulated as
[\cref{eq:ConvexLinearity} in the main text]
\begin{equation*}
	\K_{n+1}[\lambda \rho_n + (1-\lambda) \rho_n', \ldots, 
	\lambda \rho_0 + (1-\lambda) \rho_0'] = 
	\lambda \K_{n+1}[\rho_n, \ldots, \rho_0] 
	+ (1-\lambda) \K_{n+1}[\rho_n', \ldots, \rho_0'].
\end{equation*}
Since it is the same convex combination that applies to each argument, we
consider the map $\K_{n+1}$ as acting on the tuple $x = (\rho_n, \ldots, \rho_0)$ as a single variable. The set $C = \{px = (p\rho_n, \ldots, p\rho_0) \mid p \geq 0, \Tr \rho_i = 1\}$ clearly forms a convex cone. 
The set of normalized states corresponds to the base of this cone defined
by the linear functional $f(x) = \Tr \rho_n$ (or any $\Tr \rho_i$ since
$x\in C$ implies that $\Tr \rho_i = \Tr \rho_j$ for all $i,j$).
With this identification, \cref{thm:ConeExtension} applies to $\K_{n+1}$, giving a unique linear extension of $\K_{n+1}$ to the convex cone $C$ such that 
$\K_{n+1}[\lambda \rho_n, \ldots, \lambda\rho_0] = \lambda\K_{n+1}[\rho_n, \ldots, \rho_0]$.

\newpage


\section*{Supplementary material}

# Supplemental Information for "Completely-positive non-signalling non-Markovian dynamics"

Serhii Kryhin[1,*] and Vivishek Sudhir[2]

[1] *Department of Physics, Harvard University, Cambridge 02139 MA, USA*
[2] *LIGO Laboratory & Department of Mechanical Engineering, Massachusetts Institute of Technology, Cambridge 02139 MA, USA*



## CONTENTS



## I. CONTINUOUS-TIME LIMIT OF NON-MARKOVIAN EVOLUTION

We carefully consider the limit of $t_{n+1} - t_n = \Delta t \to 0$ to construct the non-Markovian master equation out of the non-Markovian analog of the Kraus theorem:

$$\mathcal{K}_{n+1}[\hat{\varrho}_n, .., \hat{\varrho}_0] = \sum_{m=0}^{n} \sum_{\alpha} V^{\dagger}_{\alpha nm} \hat{\varrho}_m V_{\alpha nm}. \tag{1}$$

We closely follow the procedure used for the Markovian case in [1], but adapt it for the infinite-dimensional case.

The assumption of the continuum limit for a discrete-time process can be written in a form

$$\begin{aligned} V_{\alpha nm} &= \Delta t \; R_\alpha(t_n, t_m) + o(\Delta t), \quad n \neq m, \\ V_{\alpha nn} &= \sqrt{\Delta t} L_\alpha(t_n) + o(\sqrt{\Delta t}), \quad \alpha \neq 0, \\ V_{0nn} &= 1 + \Delta t \; L_0(t_n) + o(\Delta t_n), \end{aligned} \tag{2}$$

where operators $R_\alpha(t_n, t_m)$ are bounded and $L_\alpha(t_n)$ are closed [2]. Then, substitution of Eq. (2) into Eq. (1) leads to

$$\begin{aligned} \hat{\varrho}_{n+1} = \hat{\varrho}_n &+ \Delta t \left[ L^{\dagger}_0(t_n) \hat{\varrho}_n + \hat{\varrho}_n L_0(t_n) \right] \\ &+ \Delta t \left[ \sum_{\alpha \neq 0} L^{\dagger}_\alpha(t_n) \hat{\varrho}_n L_\alpha(t_n) + \right] \\ + \Delta t^2 &\left[ \sum_{m=0}^{n-1} \sum_{\alpha} R_\alpha(t_n, t_m) \hat{\varrho}_m R_\alpha(t_n, t_m) \right] + o(\Delta t) \end{aligned} \tag{3}$$

The continuity assumtion of Eq. (1) results in Eq. (3) of the form

$$\hat{\varrho}_{n+1} = \hat{\varrho}_n + O(\Delta t). \tag{4}$$

so the time derivative of $\hat{\varrho}(t_n)$ in the limit of $\Delta t \to 0$ can be evaluated:

$$\frac{\partial \hat{\varrho}}{\partial t}(t) = \lim_{\substack{\Delta t \to 0 \\ t_n \to t}} \frac{\hat{\varrho}(t_{n+1}) - \hat{\varrho}(t_n)}{\Delta t}. \tag{5}$$

Substitutions of Eq. (3) into Eq. (5) leads to

$$\begin{aligned} \frac{\partial \hat{\varrho}}{\partial t} = L^{\dagger}_0(t) \hat{\varrho}(t) + \hat{\varrho}(t) L_0(t) + \sum_{\alpha \neq 0} L^{\dagger}_\alpha(t) \hat{\varrho}(t) L_\alpha(t) + \\ + \int_{t_0}^{t} dt' \sum_{\alpha} R^{\dagger}_\alpha(t, t') \hat{\varrho}(t') R_\alpha(t, t'). \end{aligned} \tag{6}$$

The trace-preserving constraint $\sum_\alpha V_{\alpha nm} V^{\dagger}_{\alpha nm} = p_{nm} \hat{I}$ with $p_{nm} > 0$ translates into

$$\sum_{\alpha} R_\alpha(t, t') R^{\dagger}_\alpha(t, t') = p(t, t') \hat{I}, \tag{7}$$

where $p(t, t') \geq 0$.

We proceed by splitting $L_0$ operators into the Hermitian and anti-Hermitian parts $L_0(t) = G(t) + iH(t)$, so both $G(t)$ and $H(t)$ are Hermitian. Eq. (6) reduces to

$$\begin{aligned} \frac{\partial \hat{\varrho}}{\partial t} = \{G(t); \hat{\varrho}(t)\} - i \left[ H(t); \hat{\varrho}(t) \right] + \\ \sum_{\alpha \neq 0} L^{\dagger}_\alpha(t) \hat{\varrho}(t) L_\alpha(t) + \int_{t_0}^{t} dt' \sum_{\alpha} R^{\dagger}_\alpha(t, t') \hat{\varrho}(t') R_\alpha(t, t'). \end{aligned} \tag{8}$$

The commutator with $H(t)$ induces trace-preserving evolution for arbitrary Hermitian $H(t)$. On the other hand, the anticommutator of $\hat{\varrho}$ with $G(t)$ in Eq. (8) is not arbitrary, but constrained by the $\text{Tr}[\partial_t \hat{\varrho}] = 0$ condition. By taking the trace on both sides of Eq. (8), one recovers

$$\begin{aligned} \int_{t_0}^{t} dt' \; \text{Tr} \Bigg[ \sum_{\alpha} R(t, t') R^{\dagger}_\alpha(t, t') \hat{\varrho}(t') + \\ + 2G(t) \hat{\varrho}(t) + \sum_{\alpha \neq 0} L_\alpha(t) L^{\dagger}_\alpha(t) \hat{\varrho}(t) \Bigg] = 0. \end{aligned} \tag{9}$$

* skryhin@g.harvard.edu

The Eq. (9) can be solved for all $\hat{\varrho}(t')$ and $\hat{\varrho}(t)$ only when Eq. (7) holds, leading to

$$G(t) = -\frac{1}{2}\int_{t_0}^{t} dt' \sum_{\alpha} R(t,t')R^{\dagger}_{\alpha}(t,t') - \\ - \frac{1}{2}\sum_{\alpha\neq 0} L_{\alpha}(t)L^{\dagger}_{\alpha}(t). \quad (10)$$

After the substitution of Eq. (10) back into Eq. (8), the resulting non-Markovian master equation attains its final form:

$$\frac{\partial\hat{\varrho}}{\partial t}(t) = (\mathcal{L}\hat{\varrho})(t) + \int_{t_0}^{t} dt'(\mathcal{R}\hat{\varrho})(t,t'). \quad (11)$$

The Markovian piece of evolution in the right side of Eq. (11) has a conventional form for the Lindblad generator:

$$(\mathcal{L}\hat{\varrho})(t) = -i[H(t);\hat{\varrho}(t)] + \\ + \sum_{\alpha\neq 0}\left(L^{\dagger}_{\alpha}(t)\hat{\varrho}(t)L_{\alpha}(t) - \frac{1}{2}\{L_{\alpha}(t)L^{\dagger}_{\alpha}(t);\hat{\varrho}(t)\}\right). \quad (12)$$

The non-Markovian contribution in the right side of Eq. (11) has a form

$$(\mathcal{R}\hat{\varrho})(t,t') = \sum_{\alpha}\left(R^{\dagger}_{\alpha}(t,t')\hat{\varrho}(t')R_{\alpha}(t,t') - \\ - \frac{1}{2}\{R_{\alpha}(t,t')R^{\dagger}_{\alpha}(t,t');\hat{\varrho}(t)\}\right). \quad (13)$$

Whenever the operators $L_{\alpha}(t)$ and $R_{\alpha}(t,t')$ can be written in a time-independent operator basis $\{F_i\}$,

$$L_{\alpha}(t) = \sum_i c_{\alpha i}(t)F_i, \quad (14)$$

$$R_{\alpha}(t) = \sum_i v_{\alpha i}(t,t')F_i, \quad (15)$$

the generators $(\mathcal{L}\hat{\varrho})(t)$ and $(\mathcal{R}\hat{\varrho})(t,t')$ can be re-written in an alternative form. The Markovian contribution becomes

$$(\mathcal{L}\hat{\varrho})(t) = -i[H(t);\hat{\varrho}(t)] + \\ + \sum_{i,j} f_{ij}(t)\left(F^{\dagger}_i\hat{\varrho}(t)F_j - \frac{1}{2}\{F_jF^{\dagger}_i;\hat{\varrho}(t)\}\right), \quad (16)$$

and the non-Markovian contribution $(\mathcal{R}\hat{\varrho})(t,t')$ becomes

$$(\mathcal{R}\hat{\varrho})(t,t') = \sum_{i,j} g_{ij}(t,t')\left(F^{\dagger}_i\hat{\varrho}(t')F_j - \frac{1}{2}\{F_jF^{\dagger}_i;\hat{\varrho}(t)\}\right). \quad (17)$$

The coefficients $f_{ij}(t)$ and $g_{ij}(t,t')$ are given by

$$f_{ij}(t) = \sum_{\alpha\neq 0} c^{*}_{\alpha i}(t)c_{\alpha j}(t), \quad (18)$$

$$g_{ij}(t,t') = \sum_{\alpha} v^{*}_{\alpha i}(t,t')v_{\alpha l}(t,t') \quad (19)$$

correspondingly.

## II. TWO-LEVEL SYSTEM COUPLED TO A BOSONIC BATH

We consider a system with the Hilbert space $\mathcal{H}_S$ and a Hamiltonian $H_S$ that is coupled to an environment with a Hilbert space $\mathcal{H}_E$ and a Hamiltonian $H_E$. The interaction Hamiltonian is assumed to be of the form

$$H_I = \sum_i s_i \otimes E_i, \quad (20)$$

where $s_i$ and $E_i$ are the operators that act on $\mathcal{H}_S$ and $\mathcal{H}_E$ correspondingly. The total Hamiltonian is given by $H_{\text{Tot}} = H_S + H_I + H_E$. The master equation for the unitary evolution of a total density matrix $\tilde{\varrho}_{\text{Tot}}$ in the interaction picture takes the form

$$\partial_t\tilde{\varrho}_{\text{Tot}}(t) = -i\left[H_I(t);\tilde{\varrho}_{\text{Tot}}(t)\right]. \quad (21)$$

We distinguish the density matrix $\tilde{\varrho}_{\text{Tot}}$ in the interaction picture from $\hat{\varrho}_{\text{Tot}}$ in the Schrodinger picture by denoting it with a tilde instead of a circumflex.

To derive the non-Markovian master equation for the system in question, follow the Born-Markov approximation [3] but reject the assumption of Markovianity in the process. Formal integration of Eq. (21) leads to

$$\tilde{\varrho}_{\text{Tot}}(t) = \tilde{\varrho}_{\text{Tot}}(t_0) - i\int_0^t ds\ \left[H_I(s);\tilde{\varrho}_{\text{Tot}}(s)\right]. \quad (22)$$

Substitution of Eq. (22) back into Eq. (21) results in

$$\partial_t\tilde{\varrho}_{\text{Tot}}(t) = -i\left[H_I(t);\tilde{\varrho}_{\text{Tot}}(t_0)\right] - \\ - \int_{t_0}^{t} ds\ \left[H_I(t);\left[H_I(s);\tilde{\varrho}_{\text{Tot}}(s)\right]\right]. \quad (23)$$

Our first assumption is to neglect the correlations between the system and the bath, so the joint density matrix can be approximated as

$$\tilde{\varrho}_{\text{Tot}}(t) = \tilde{\varrho}(t)\otimes\tilde{\varrho}_E(t). \quad (24)$$

States $\tilde{\varrho}(t)$ and $\tilde{\varrho}_E(t)$ are reduced density matrices of the system and environment, respectively. Our second assumption is to neglect the back-action of the system on the environment, so $\tilde{\varrho}_E(t) = \tilde{\varrho}_E$ in the interaction picture. Without loss of generality, we additionally assume that $\langle E(t)\rangle = 0$. Under the assumptions stated above, the partial trace over the environment degrees of freedom in Eq. (23) leads to

$$\partial_t\tilde{\varrho} = -\int_{t_0}^{t} ds\ \sum_{ij}\left([s_i(t);s_j(s)\tilde{\varrho}(s)]\langle E_i(t)E_j(s)\rangle - \\ - [s_i(t);\tilde{\varrho}(s)s_j(s)]\langle E_j(s)E_i(t)\rangle\right). \quad (25)$$

This is the point where we diverge from the Born-Markov approximation [3]: we will not add the Markov approximation to remove the integral over time in the right side of Eq. (25).

Instead, we focus on a concrete example and consider a two-level system (TLS). The states of the system are enumerated as $|1\rangle$ and $|2\rangle$ and are split by the energy $\omega_0$. The Hamiltonian of the system is assumed to be $H = \omega_0\,\sigma_z/2$. We consider the coupling of the TLS to a bosonic bath with the Hamiltonian

$$H_E = \int d\omega\, \omega\, b^\dagger(\omega)b(\omega), \quad [b(\omega), b^\dagger(\omega')] = \delta(\omega-\omega'). \tag{26}$$

The relevant bath operator $E(t)$ is set to have a form

$$E(t) = g\int_0^\infty d\omega\; K(\omega)b(\omega)e^{-i\omega t}, \tag{27}$$

where $g$ is a small coupling constant and a $K(\omega)$ is a weight of the coupling of different modes of the bath to the TLS. The interaction Hamiltonian of the TLS and a bath is $H_I(t) = \sigma_{21}^\dagger(t)E(t) + \sigma_{21}(t)E^\dagger(t)$. In the expressions above we defined $\sigma_z = |1\rangle\langle 1| - |2\rangle\langle 2|$ and $\sigma_{21} = |2\rangle\langle 1|$. We assume that the bosonic bath is in the stationary state, so its density matrix commutes with the Hamiltonian, and $\mathrm{Tr}\left[b^\dagger(\omega)b(\omega')\varrho_E\right] = n(\omega)\delta(\omega-\omega')$.

Transforming Eq. (25) back into the Schrödinger picture results under the above assumptions for the TLS and the bosonic bath results in

$$\begin{aligned}\partial_t\hat\varrho(t) &= -i[H_S;\hat\varrho(t)] + ig^2\int_{t_0}^t ds\,[\sigma_{22};U_{t,s}[\hat\varrho]]G_E(t-s) +\\ &+ g^2\int_{t_0}^t ds\,\left(2\sigma_{21}^\dagger U_{t,s}[\hat\varrho]\sigma_{21} - \{\sigma_{21}\sigma_{21}^\dagger; U_{t,s}[\hat\varrho]\}\right)S_E(t-s)\\ &+ g^2\int_{t_0}^t ds\,\left(2\sigma_{21}U_{t,s}[\hat\varrho]\sigma_{21}^\dagger - \{\sigma_{21}^\dagger\sigma_{21}; U_{t,s}[\hat\varrho]\}\right)\bar S_E(t-s),\end{aligned} \tag{28}$$

where $\sigma_{22} = |2\rangle\langle 2|$, and

$$U_{t,s}[\hat\varrho] = e^{iH_S(t-s)}\hat\varrho(s)e^{-iH_S(t-s)}, \tag{29}$$

$$S_E(t) = \int_0^\infty d\omega\,|K(\omega)|^2 n(\omega)\cos((\omega-\omega_0)t), \tag{30}$$

$$\bar S_E(t) = \int_0^\infty d\omega\,|K(\omega)|^2(n(\omega)+1)\cos((\omega-\omega_0)t), \tag{31}$$

$$G_E(t) = \int_0^\infty d\omega\,|K(\omega)|^2\sin((\omega-\omega_0)t). \tag{32}$$

We proceed by analyzing the difference between $\hat\varrho(t)$ and $U_{t,s}[\hat\varrho]$. Since the coupling between the TLS and the bath is weak, the difference between these two quantities is $O(g^2)$. Since the terms in Eq. (28), where $U_{t,s}[\hat\varrho]$ appears, are already proportional to $g^2$, the difference between $\hat\varrho(t)$ and $U_{t,s}[\hat\varrho]$ can be neglected in the weak-interaction approximation in the anti-commutator terms of Eq. (28). With a help of relation $\sigma_{21}e^{iH_St} = e^{i\omega_0 t/2}\sigma_{21}$, the Master equation can then be written as

$$\begin{aligned}\partial_t\hat\varrho(t) &= -i[H_S;\hat\varrho(t)] + ig^2\int_{t_0}^t ds\,[\sigma_{22};\hat\varrho(t)]G_E(t-s) +\\ &+ g^2\int_{t_0}^t ds\,\left(2\sigma_{21}^\dagger\hat\varrho(s)\sigma_{21} - \{\sigma_{21}\sigma_{21}^\dagger;\hat\varrho(t)\}\right)S_E(t-s)\\ &+ g^2\int_{t_0}^t ds\,\left(2\sigma_{12}^\dagger\hat\varrho(s)\sigma_{12} - \{\sigma_{12}\sigma_{12}^\dagger;\hat\varrho(t)\}\right)\bar S_E(t-s),\end{aligned} \tag{33}$$

The term that involves $G_E(t-s)$ is of a Stark shift type, and therefore can be re-incorporated into the definition of $H_S$. After neglecting the quantum correction that differentiates between $S_E(t)$ and $\bar S_E(t)$, the master equation becomes

$$\begin{aligned}\partial_t\hat\varrho(t) &= -i[H_S;\hat\varrho(t)] +\\ &+ \int_{t_0}^t ds\, g_{ij}(t-s)\left(F_j^\dagger\hat\varrho(s)F_i - \frac{1}{2}\{F_iF_j^\dagger,\ \hat\varrho(t)\}\right).\end{aligned} \tag{34}$$

The matrix $g_{ij}(t,t')$ in the basis given by $\{F_i\} = \{\sigma_{12}, \sigma_{21}, \sigma_z/\sqrt{2}, I/\sqrt{2}\}$ is given by

$$[g(t)]_{ij} = \begin{bmatrix} 2g^2S_E(t) & 0 & 0 & 0\\ 0 & 2g^2S_E(t) & 0 & 0\\ 0 & 0 & 0 & 0\\ 0 & 0 & 0 & 0\end{bmatrix}. \tag{35}$$

Note that $S_E(t) \geq 0$ due to Bochner's theorem, so $g_{ij}(t)$ is a positive semi-definite matrix. The trace-preserving constraint becomes $\sum_{ij} g_{ij}(t)F_jF_i^\dagger = 2g^2S(t) \geq 0$, so the trace is preserved under the evolution governed by Eq. (34). Eq. (34) is exactly of the form of the non-Markovian master equation given in the main text.

The derivation conducted above is in the spirit of the Markov-Born approximation. In fact, it differs from the Markov-Born approximation only by rejecting the Markovianity assumption. If one assumes

$$|K(\omega)|^2 n(\omega) \approx \gamma/2\pi = \text{const}, \tag{36}$$

$$|K(\omega)|^2(n(\omega)+1) \approx \bar\gamma/2\pi = \text{const}, \tag{37}$$

the $S_E(t)$ and $\bar S_E(t)$ immediately reduce to $S_E(t) \approx \gamma\delta(t)$ and $\bar S_E(t) \approx \bar\gamma\delta(t)$, turning evolution Markovian. Therefore, Eq. (34) shall be regarded as a natural generalization of the well-accepted result in the literature for a case of a more sophisticated bosonic bath.

### A. Evolution of a TLS coupled to non-Markovian bath

To resolve the physical consequences of non-Markovian dynamics, we consider a two-level system that obeys the equation of the form

$$\frac{\partial\hat\varrho}{\partial t}(t) = (\mathcal{L}\hat\varrho)(t) + \int_{-\infty}^t dt'(\mathcal{R}\hat\varrho)(t,t'), \tag{38}$$

where the Hamiltonian is $H = \omega_0\sigma_z/2 + \Omega e^{i\nu t}\sigma_{21} + \text{h.c.}$, the Lindbladian operator that is responsible for the Markovian part of the evolution is

$$(\mathcal{L}\hat{\varrho})(t) = -i[H(t);\hat{\varrho}(t)] + \\ + \Gamma_{\rm M}\left(\sigma_{21}^\dagger\hat{\varrho}(t)\sigma_{21} - \frac{1}{2}\{\sigma_{21}\sigma_{21}^\dagger;\ \hat{\varrho}(t)\}\right) + \\ + \bar{\Gamma}_{\rm M}\left(\sigma_{12}^\dagger\hat{\varrho}(t)\sigma_{12} - \frac{1}{2}\{\sigma_{12}\sigma_{12}^\dagger;\hat{\varrho}(t)\}\right), \quad (39)$$

and the new, non-Markovian contribution $(\mathcal{R}\hat{\varrho})(t,t')$ is

$$(\mathcal{R}\hat{\varrho})(t,t') = S_E(t-t')\left(\sigma_{21}^\dagger\hat{\varrho}(t')\sigma_{21} - \frac{1}{2}\{\sigma_{21}\sigma_{21}^\dagger;\hat{\varrho}(t)\}\right) + \\ + S_E(t-t')\left(\sigma_{12}^\dagger\hat{\varrho}(t')\sigma_{12} - \frac{1}{2}\{\sigma_{12}\sigma_{12}^\dagger;\hat{\varrho}(t)\}\right). \quad (40)$$

The Markovian part of the evolution corresponds to a two-level system with a gap $\Delta$ driven by a classical drive $E(t)$. The two-level system is coupled to a Markovian bath that induces decoherence with rates $\Gamma_{\rm M}$ and $\bar{\Gamma}_{\rm M}$. The non-Markovian part of the bath $(\mathcal{R}\hat{\varrho})(t,t')$ is described by the bath correlation function $S(t-t')$.

## B. Continuous measurement

The observables that do not reduce to single-time quantities and are of practical interest are the continuous measurement observables. We explore the effects of continuous measurement on the two-level system coupled to a non-Markovian bath. As was shown in the main text, the measured correlation function that corresponds to observable $X$ is given by

$$\mathbb{E}[I_X(t+\tau)I_X(t)] = \text{Tr}\left[(\mathcal{X}\hat{\varrho}_{X,t})(t+\tau)\right], \quad (41)$$

where $I_X$ is the diffusion current that corresponds to obstrvable $X$ (see main text), the map $(\mathcal{X}A)(t) \equiv XA(t) + A(t)X$ is defined in Schrödinger picture, and qunatity $\hat{\varrho}_{X,t}(t+\tau)$ is defined by

$$\partial_\tau\hat{\varrho}_{X,t}(t+\tau) = \mathcal{L}[\varrho_{X,t}](t+\tau) \\ + \int_t^{t+\tau} dt'\mathcal{R}[\hat{\varrho}_{X,t}](t+\tau,t') \\ + \mathbb{E}[I_X(t)]\int_{-\infty}^t dt'\,\mathcal{R}[\hat{\varrho}](t+\tau,t'), \quad (42)$$

as dictated by Eq. (15) of the main text. The quantity $\mathbb{E}[I_X(t)] = \text{Tr}[(\mathcal{X}\hat{\varrho})(t)]$ is the expectation of $X$ at time $t$. The initial conditions that define $\hat{\varrho}_{X,t}(t)$ are such that $\hat{\varrho}_{X,t}(t) = (\mathcal{X}\hat{\varrho})(t) = X\hat{\varrho}(t) + \hat{\varrho}(t)X^\dagger$.

Focusing on a driven two-level system described by Eqs. (39) and (40) and the same Hamiltonian as in Section II A, the explicit equations of motion for the components $\hat{\varrho}_{X,t,ij}(t+\tau) = \langle i|\hat{\varrho}_{X,t}(t+\tau)|j\rangle$ take form

$$\partial_\tau\varrho_{X,t,11}(t+\tau) = i\Omega e^{i\nu(t+\tau)}\varrho_{X,t,12}(t+\tau) - i\Omega^* e^{-i\nu(t+\tau)}\varrho_{X,t,21}(t+\tau) + \Gamma_{\rm M}\varrho_{X,t,22}(t+\tau) - \\ -(\bar{\Gamma}_{\rm M} + \Gamma_{\rm NM})\varrho_{X,t,11}(t+\tau) + \int_t^{t+\tau} dt' S_E(t+\tau-t')\varrho_{X,t,22}(t') + \mathbb{E}[I_X(t)]\int_{-\infty}^t dt' S_E(t+\tau-t')\varrho_{22}(t'), \quad (43)$$

$$\partial_\tau\varrho_{X,t,22}(t+\tau) = i\Omega^* e^{-i\nu(t+\tau)}\varrho_{X,t,21}(t+\tau) - i\Omega e^{i\nu(t+\tau)}\varrho_{X,t,12}(t+\tau) + \bar{\Gamma}_{\rm M}\varrho_{X,t,11}(t+\tau) - \\ -(\Gamma_{\rm M} + \Gamma_{\rm NM})\varrho_{X,t,22}(t+\tau) + \int_t^{t+\tau} dt' S_E(t+\tau-t')\varrho_{X,t,11}(t') + \mathbb{E}[I_X(t)]\int_{-\infty}^t dt' S_E(t+\tau-t')\varrho_{11}(t'), \quad (44)$$

$$\partial_\tau\varrho_{X,t,12}(t+\tau) = -i\omega_0\varrho_{X,t,12}(t+\tau) - iE^*(t+\tau)\left(\varrho_{X,t,22}(t+\tau) - \varrho_{X,t,11}(t+\tau)\right) - \\ - \frac{1}{2}(\Gamma_{\rm M} + \bar{\Gamma}_{\rm M} + 2\Gamma_{\rm NM})\varrho_{X,t,12}(t+\tau). \quad (45)$$

The rate $\Gamma_{\rm NM}$ is given by $\Gamma_{\rm NM} = \int_0^\infty S_E(t)$. We focus on the case of the quasi-static pumping rate $\Omega = \text{const}$, and the polarization observable $X = \sigma_{12}$, so $\mathbb{E}[I_X(t)] = 2\,\text{Re}[\sigma_{12}\hat{\varrho}(t)]$. In the rotating frame,

$\varrho_{X,t,12}(t+\tau) = \varrho'_{X,t,12}(t+\tau)e^{-i\nu(t+\tau)}$, and the Eqs. (43) to (45) simplify to

$$\partial_\tau \varrho_{X,t,11}(t+\tau) = 2\,\mathrm{Re}\left[i\Omega\varrho'_{X,t,12}(t+\tau)\right] + \Gamma_{\mathrm{M}}\varrho_{X,t,22}(t+\tau) - (\bar{\Gamma}_{\mathrm{M}} + \Gamma_{\mathrm{NM}})\varrho_{X,t,11}(t+\tau) + \\ + \int_t^{t+\tau} dt' S_E(t+\tau-t')\varrho_{X,t,22}(t') + \mathbb{E}[I_X(t)]\int_{-\infty}^{t} dt' S_E(t+\tau-t')\varrho_{22}(t'), \quad (46)$$

$$\partial_\tau \varrho_{X,t,22}(t+\tau) = 2\,\mathrm{Re}\left[i\Omega^*\varrho'_{X,t,21}(t+\tau)\right] + \bar{\Gamma}_{\mathrm{M}}\varrho_{X,t,11}(t+\tau) - (\Gamma_{\mathrm{M}} + \Gamma_{\mathrm{NM}})\varrho_{X,t,22}(t+\tau) + \\ + \int_t^{t+\tau} dt' S_E(t+\tau-t')\varrho_{X,t,11}(t') + \mathbb{E}[I_X(t)]\int_{-\infty}^{t} dt' S_E(t+\tau-t')\varrho_{11}(t'), \quad (47)$$

$$\partial_\tau \varrho'_{X,t,12}(t+\tau) = -(\Gamma_{\mathrm{T}} + i\Delta)\varrho'_{X,t,12}(t+\tau) - i\Omega^*\left(\varrho_{X,t,22}(t+\tau) - \varrho_{X,t,11}(t+\tau)\right). \quad (48)$$

In the language of the rotating frame, the initial conditions for the system of equations above become $\varrho_{X,t,11}(t) = \varrho'_{12}(t)e^{-i\nu t} + \varrho'_{21}(t)e^{i\nu t}$, $\varrho'_{X,t,12}(t) = \varrho_{22}(t)e^{i\nu t}$, and $\varrho_{X,t,22} = 0$. To solve the Eqs. (46) to (48) with the listed initial conditions, we adopt the convention that $\hat{\varrho}_{X,t}(t') = 0$ for $t' < t$, so the Eqs. (46) to (48) become a system of Volterra integro-differential equations.

The solution of Eqs. (46) to (48) is constructed in a more conscise notation: $h_{\mathrm{R}}(\tau) = \mathrm{Re}[\varrho'_{X,t,12}(t+\tau)]$, $h_{\mathrm{I}}(\tau) = \mathrm{Im}[\varrho'_{X,t,12}(t+\tau)]$, $f_\pm(\tau) = \varrho_{X,t,22}(t+\tau) \pm \varrho_{X,t,11}(t+\tau)$. In the new variables, the equations of motion simplify to

$$\partial_\tau f_+(\tau) = -\Gamma_{\mathrm{NM}} + \int_0^\tau d\tau'\, S_E(\tau-\tau')f_+(\tau') + F_+(\tau), \quad (49)$$

$$\partial_\tau f_-(\tau) = 4\Omega h_{\mathrm{I}}(\tau) - (\Gamma_{\mathrm{M}} + \bar{\Gamma}_{\mathrm{M}} + \bar{\Gamma}_{\mathrm{NM}})f_-(\tau) - \int_0^\tau d\tau'\, S_E(\tau-\tau')f_-(\tau') + F(\tau), \quad (50)$$

$$\partial_\tau h_{\mathrm{I}}(\tau) = -\Gamma_{\mathrm{T}}h_{\mathrm{I}}(\tau) + \Delta h_{\mathrm{R}}(\tau) - \Omega f_-(\tau), \quad (51)$$

$$\partial_\tau h_{\mathrm{R}}(\tau) = -\Gamma_{\mathrm{T}}h_{\mathrm{R}}(\tau) - \Delta h_{\mathrm{I}}(\tau). \quad (52)$$

It was assumed above that $\Omega$ is real. The functions $F_+(\tau)$ and $F(\tau)$ can be expressed through $S(\tau)$ and $f_+(\tau)$ as

$$F_+(\tau) = \mathbb{E}[I_X(t)]\int_0^{+\infty} d\tau' S_E(\tau+\tau'), \quad (53)$$

$$F(\tau) = \mathbb{E}[I_X(t)]\int_0^{+\infty} d\tau' S_E(\tau+\tau')\left(\varrho_{11}(t-\tau') - \varrho_{22}(t-\tau')\right) + (\bar{\Gamma}_{\mathrm{M}} - \Gamma_{\mathrm{M}})f_+(\tau). \quad (54)$$

Eq. (49) can be solved separately from Eqs. (50) to (52), since it is independent from the rest of the degrees of freedom. The Wiener-Hopf factorization method [4] yields a solution that in the Fourier space to takes the form

$$\tilde{f}_+(\omega) = \frac{f_+(0) + \tilde{F}_+(\omega)}{-i\omega + \Gamma_{\mathrm{NM}} - \gamma_{\mathrm{NM}}(\omega) + \varepsilon}, \quad (55)$$

where, as before, $\gamma_{\mathrm{NM}}(\omega) = \int_0^\infty dt\, e^{i\omega t} S_E(t)$, $\Gamma_{\mathrm{NM}} = \gamma_{\mathrm{NM}}(0)$, $\tilde{F}_+(\omega) = \int_0^\infty d\tau\, e^{i\omega\tau} F(\tau)$, and $\varepsilon$ is an infinitesimal positive quantity. The real-time quantity $f_+(\tau)$ can be reconstructed with

$$f_+(\tau) = \int_{-\infty}^{\infty} \frac{d\omega}{2\pi} e^{-i\omega t} f_+(\omega). \quad (56)$$

Expression in Eq. (55) is valid as far as $-i\omega + \Gamma_{\mathrm{NM}} - \gamma_{\mathrm{NM}}(\omega) + \varepsilon$ does not have zeros in the upper half-plane. Performing the Fourier transform on the Eqs. (50) to (52) yields

$$\begin{bmatrix} -i\omega + 2\gamma_{\mathrm{T}}(\omega) & -4\Omega & 0 \\ \Omega & -i\omega + \Gamma_{\mathrm{T}} & -\Delta \\ 0 & \Delta & -i\omega + \Gamma_{\mathrm{T}} \end{bmatrix} \begin{bmatrix} \tilde{f}_-(\omega) \\ \tilde{h}_{\mathrm{I}}(\omega) \\ \tilde{h}_{\mathrm{R}}(\omega) \end{bmatrix} = \begin{bmatrix} \tilde{F}(\omega) + f_-(0) \\ h_{\mathrm{I}}(0) \\ h_{\mathrm{R}}(0) \end{bmatrix} \quad (57)$$

where $\tilde{f}_-(\omega)$, $\tilde{h}_{\mathrm{I}}(\omega)$, and $\tilde{h}_{\mathrm{R}}(\omega)$ are Fourier transforms of $f_-(\tau)$, $h_{\mathrm{I}}(\tau)$, and $h_{\mathrm{R}}(\tau)$ correspondingly, $\gamma_{\mathrm{T}}(\omega) = (\Gamma_{\mathrm{M}} + \bar{\Gamma}_{\mathrm{M}} + \Gamma_{\mathrm{NM}} + \gamma_{\mathrm{NM}}(\omega))/2$, $\Gamma_{\mathrm{T}} = \gamma_{\mathrm{T}}(0)$, and $\tilde{F}(\omega) = \int_0^\infty d\tau\, e^{i\omega t} F(\tau)$. The matrix at the left side of Eq. (57) attains simple expressions for its eigenvalues when $\Omega, \Delta \gg \Gamma_{\mathrm{T}}$:

$$\lambda_0 \approx -i\omega + \Gamma_{\mathrm{T}} + \frac{\Delta^2}{\Omega_{\mathrm{R}}^2}\gamma''_{\mathrm{T}}(\omega), \quad \lambda_\pm \approx -i\omega + \Gamma_{\mathrm{T}} \pm i\Omega_{\mathrm{R}} + \frac{2\Omega^2}{\Omega_{\mathrm{R}}^2}\gamma''_{\mathrm{T}}(\omega). \quad (58)$$

The characteristic frequency $\Omega_{\rm R}$ is defined as $\Omega_{\rm R} = \sqrt{4\Omega^2+\Delta^2}$, and the rate $\gamma''_{\rm T}(\omega)$ is expressed as $\gamma''_{\rm T}(\omega) = (\Gamma_{\rm M}+\bar{\Gamma}_{\rm M}+2\gamma_{\rm NM}(\omega))/2$. Inverting the matrix in the left side of Eq. (57) under the assumptions of $\Omega, \Delta \gg \Gamma_{\rm T}$ results in the following expressions for $\tilde{f}_-(\omega)$, $\tilde{h}_{\rm R}(\omega)$, and $\tilde{h}_I(\omega)$:

$$\tilde{f}_-(\omega) \approx \frac{1}{\lambda_0}\left[\frac{\Delta^2}{\Omega_{\rm R}^2}\left(\tilde{F}(\omega)+f_-(0)\right)+\frac{4\Delta\Omega}{\Omega_{\rm R}^2}h_{\rm R}(0)\right]+\frac{1}{\lambda_+}\left[\frac{2\Omega^2}{\Omega_{\rm R}^2}\left(\tilde{F}(\omega)+f_-(0)\right)+\frac{2i\Omega}{\Omega_{\rm R}}h_{\rm I}(0)-\frac{2\Omega\Delta}{\Omega_{\rm R}^2}h_{\rm R}(0)\right] + \frac{1}{\lambda_-}\left[\frac{2\Omega^2}{\Omega_{\rm R}^2}\left(\tilde{F}(\omega)+f_-(0)\right)-\frac{2i\Omega}{\Omega_{\rm R}}h_{\rm I}(0)-\frac{2\Omega\Delta}{\Omega_{\rm R}^2}h_{\rm R}(0)\right] \tag{59}$$

$$\tilde{h}_{\rm R}(\omega) \approx \frac{1}{\lambda_0}\left[\frac{\Omega\Delta}{\Omega_{\rm R}^2}\left(\tilde{F}(\omega)+f_-(0)\right)+\frac{4\Omega^2}{\Omega_{\rm R}^2}h_{\rm R}(0)\right]+\frac{1}{\lambda_+}\left[-\frac{\Omega\Delta}{2\Omega_{\rm R}^2}\left(\tilde{F}(\omega)+f_-(0)\right)-\frac{i\Delta}{2\Omega_{\rm R}}h_{\rm I}(0)+\frac{\Delta^2}{2\Omega_{\rm R}^2}h_{\rm R}(0)\right] + \frac{1}{\lambda_-}\left[-\frac{\Omega\Delta}{2\Omega_{\rm R}^2}\left(\tilde{F}(\omega)+f_-(0)\right)+\frac{i\Delta}{2\Omega_{\rm R}}h_{\rm I}(0)+\frac{\Delta^2}{2\Omega_{\rm R}^2}h_{\rm R}(0)\right] \tag{60}$$

$$\tilde{h}_{\rm I}(\omega) \approx \frac{i}{\lambda_+}\left[-\frac{\Omega}{2\Omega_{\rm R}}\left(\tilde{F}(\omega)+f_-(0)\right)-\frac{i}{2}h_{\rm I}(0)+\frac{\Delta}{2\Omega_{\rm R}}h_{\rm R}(0)\right] - - \frac{i}{\lambda_-}\left[-\frac{\Omega}{2\Omega_{\rm R}}\left(\tilde{F}(\omega)+f_-(0)\right)+\frac{i}{2}h_{\rm I}(0)+\frac{\Delta}{2\Omega_{\rm R}}h_{\rm R}(0)\right] \tag{61}$$

The quantity of interest $\tilde{\varrho}'_{X,t,12}(\omega)$, which is a Fourier transform of $\varrho'_{X,t,12}(t+\tau)$ over $\tau$, can be extracted from Eqs. (60) and (61), leading to

$$\tilde{\varrho}'_{X,t,12}(\omega) \approx \tilde{h}_{\rm R}(\omega)+i\tilde{h}_{\rm I}(\omega) = \frac{1}{\lambda_0}\left[\frac{\Omega\Delta}{\Omega_{\rm R}^2}\left(\tilde{F}(\omega)+f_-(0)\right)+\frac{4\Omega^2}{\Omega_{\rm R}^2}h_{\rm R}(0)\right]+ + \frac{1}{\lambda_+}\left(\frac{\Delta}{\Omega_{\rm R}}-1\right)\left[-\frac{\Omega}{2\Omega_{\rm R}}\left(\tilde{F}(\omega)+f_-(0)\right)-\frac{i}{2}h_{\rm I}(0)+\frac{\Delta}{2\Omega_{\rm R}}h_{\rm R}(0)\right]+ + \frac{1}{\lambda_-}\left(\frac{\Delta}{\Omega_{\rm R}}+1\right)\left[-\frac{\Omega}{2\Omega_{\rm R}}\left(\tilde{F}(\omega)+f_-(0)\right)+\frac{i}{2}h_{\rm I}(0)+\frac{\Delta}{2\Omega_{\rm R}}h_{\rm R}(0)\right] \tag{62}$$

For simplicity, we consider a limit of small detuning $\Delta \ll \Omega, \Omega_{\rm R}$. In this limit $\Omega_{\rm R} \approx 2\Omega$ and $\bar{\varrho}_{11} \approx \bar{\varrho}_{22} \approx 1/2$, so $f_\pm(0) \approx 0$ and $\tilde{F}(\omega) \approx 0$. The expression for $\tilde{\varrho}_{X,t,12}(\omega)$ further simplifies to

$$\tilde{\varrho}'_{X,t,12}(\omega) \approx \frac{1}{\lambda_0}h_{\rm R}(0)+\frac{i\,h_{\rm I}(0)}{2}\left[\frac{1}{\lambda_+}+\frac{1}{\lambda_-}\right]. \tag{63}$$

With the expression for $\tilde{\varrho}_{X,t,12}(\omega)$ in hand, we compute the correlation function of the associated diffusion current $\mathbb{E}[I_X(t+\tau)I_X(t)]$. The diffusion current correlation function that corresponds to the polarization $X=\sigma_{12}$, according to Eq. (41), can be expressed as

$$\mathbb{E}[I_X(t+\tau)I_X(t)] = 2\,{\rm Re}[\varrho_{X,t,21}(t+\tau)] = 2\,{\rm Re}\left[e^{-i\nu(t+\tau)}\varrho'_{X,t,12}(t+\tau)\right]. \tag{64}$$

The power spectral density of $E[I_X(t+\tau)I_X(t)]$ then can be expressed as

$$S(\omega) \equiv \int_0^{+\infty} d\tau\, \cos(\omega\tau)\mathbb{E}[I_X(t+\tau)I_X(t)] = {\rm Re}\left[e^{-i\nu t}\tilde{\varrho}'_{X,t,12}(\omega-\nu)\right]. \tag{65}$$

After keeping only the time-translation-invariant part of the argument of the right side in Eq. (65), one recovers the following expression for $S(\omega)$:

$$S(\omega+\omega_0) = {\rm Re}\left[\frac{1}{4\lambda_0}+\frac{1}{8\lambda_+}+\frac{1}{8\lambda_-}\right]. \tag{66}$$

This is an analog of the Mollow triple [5] in the case of the non-Markovian evolution.



\end{document}